\documentclass{article}
\usepackage{graphicx,epsfig,color}
\usepackage{latexsym}

\newtheorem{theorem}{Theorem}

\newtheorem{corollary}[theorem]{Corollary}
\newtheorem{definition}[theorem]{Definition}

\newcommand{\scrod}{\quad\nopagebreak}

\newenvironment{proof}
{\bigskip\noindent\textbf{Proof~}} {\marginpar{$\Box$}\bigskip}

\begin{document}

\date{}

\title{The Complexity of Testing Monomials in Multivariate Polynomials}

\author{Zhixiang Chen and Bin Fu
 \\ \\
Department of Computer Science\\
 University of Texas-Pan American\\
 Edinburg, TX 78539, USA\\
\{chen,binfu\}@cs.panam.edu\\\\
} \maketitle

\begin{abstract}
The work in this paper is to initiate a theory of testing
monomials in multivariate polynomials. The central question is to
ask whether a polynomial represented by certain economically
compact structure has a multilinear monomial in its sum-product
expansion. The complexity aspects of this problem and its variants
are investigated with two folds of objectives. One is to
understand how this problem relates to critical problems in
complexity, and if so to what extent. The other is to exploit
possibilities of applying algebraic properties of polynomials to
the study of those problems. A series of results about
$\Pi\Sigma\Pi$ and $\Pi\Sigma$ polynomials are obtained in this
paper, laying a basis for further study along this line.
\end{abstract}

\section{Introduction}

We begin with two examples to exhibit the motivation and necessity
of the study about the monomial testing problem for multivariate
polynomials. The first is about testing a  $k$-path in any given
undirected graph $G=(V,E)$ with $|V| = n$, and the second is about
the satisfiability problem. Throughout this paper, polynomials
refer to those with multiple variables.

For any fixed integer $c\ge 1$, for each vertex $v_i \in V$,
define a polynomial $p_{k,i}$ as follows:
\begin{eqnarray}
p_{1,i}  &=&  x_i^c, \nonumber \\
p_{k+1,i} &=&  x_i^c   \left(\sum_{(v_i,v_j)\in E} p_{k,j}\right),
\ k
>1. \nonumber
\end{eqnarray}
We define a polynomial for $G$ as
\begin{eqnarray}
p(G, k)  &=&  \sum^{n}_{i=1} p_{k,i}. \nonumber
\end{eqnarray}
Obviously, $p(G,k)$ can be represented by an arithmetic circuit.
It is easy to see that the graph $G$ has a $k$-path $v_{i_1}\cdots
v_{i_k}$ iff $p(G, k)$ has a monomial of $x_{i_1}^c\cdots
x_{i_k}^c$ of degree $ck$ in its sum-product expansion. $G$ has a
Hamiltonian path iff $p(G, n)$ has the monomial $x_1^c\cdots
x_n^c$ of degree $cn$ in its sum-product expansion. One can also
see that a path with some loop can be characterized by a monomial
as well. Those observations show that testing monomials in
polynomials is closely related to solving $k$-path, Hamiltonian
path and other problems about graphs. When $c=1$, $x_{i_1}\cdots
x_{i_k}$ is multilinear. The problem of testing multilinear
monomials has recently been exploited by Koutis \cite{koutis08}
and Williams \cite{williams09} to design innovative randomized
parameterized algorithms for the $k$-path problem.

Now, consider any CNF formula $f= f_1 \wedge \cdots \wedge f_m$, a
conjunction of $m$ clauses with each clause $f_i$ being a
disjunction of some variables or negated ones. We may view
conjunction as multiplication and disjunction as addition, so $f$
looks like a {\em "polynomial"}, denoted by $p(f)$. $p(f)$ has a
much simpler $\Pi\Sigma$ representation, as will be defined in the
next section, than general arithmetic circuits. Each {\em
"monomial"} $\pi = \pi_1 \ldots \pi_m$ in the sum-product
expansion of $p(f)$ has a literal $\pi_i$ from the clause $f_i$.
Notice that  a boolean variable $x \in Z_2$ has two properties of
$x^2 = x$ and $x \bar{x} = 0$. If we could realize these
properties for $p(f)$ without unfolding it into its sum-product,
then $p(f)$ would be a {\em "real polynomial"} with two
characteristics: (1) If $f$ is satisfiable then $p(f)$ has a
multilinear monomial, and (2) if $f$ is not satisfiable then
$p(f)$ is identical to zero. These would give us two approaches
towards testing the satisfiability of $f$. The first is to test
multilinear monomials in $p(f)$, while the second is to test the
zero identity of $p(f)$. However, the task of realizing these two
properties with some algebra to help transform $f$ into a needed
polynomial $p(f)$ seems, if not impossible, not easy. Techniques
like arithmetization in Shamir \cite{shamir92} may not be suitable
in this situation. In many cases, we would like to move from $Z_2$
to some larger algebra so that we can enjoy more freedom to use
techniques that may not be available when the domain is too
constrained. The algebraic approach within $Z_2[Z^k_2]$ in Koutis
\cite{koutis08} and Williams \cite{williams09} is one example
along the above line. It was proved in Bshouty {\em et al.}
\cite{bshouty95}  that extensions of  DNF formulas over $Z^n_2$ to
$Z_N$-DNF formulas over the ring $Z^n_N$ are learnable by a
randomized algorithm with equivalence queries, when $N$ is large
enough. This is possible because a larger domain may allow more
room to utilize randomization.

There has been a long history in complexity theory with heavy
involvement of studies and applications of polynomials. Most
notably, low degree polynomial testing/representing and polynomial
identity testing have played invaluable roles in many major
breakthroughs in complexity theory. For example, low degree
polynomial testing is involved in the proof of the PCP Theorem,
the cornerstone of the theory of computational hardness of
approximation and the culmination of a long line of research on IP
and PCP (see, Arora {\em at el.} \cite{arora98} and Feige {\em et
al.} \cite{feige96}). Polynomial identity testing has been
extensively studied due to its role in various aspects of
theoretical computer science (see, for examples, Chen and Kao
\cite{chen00}, Kabanets and Impagliazzo \cite{kabanets03}) and its
applications in various fundamental results such as Shamir's
IP=PSPACE \cite{shamir92} and the AKS Primality Testing
\cite{aks04}. Low degree polynomial representing
\cite{minsky-papert68} has been sought for so as to prove
important results in circuit complexity, complexity class
separation and subexponential time learning of boolean functions
(see, for examples, Beigel \cite{beigel93}, Fu\cite{fu92} and
Klivans and Servedio \cite{klivans01}). These are just a few
examples. A survey of the related literature is certainly beyond
the scope of this paper.

The above two examples of the $k$-path testing and satisfiability
problems, the rich literature about polynomial testing and many
other observations have motivated us to develop a new theory of
testing monomials in polynomials represented by economically
compact structures. The monomial testing problem is related to,
and somehow complements with, the low degree testing and the
identity testing of polynomials. We want to investigate various
complexity aspects of the monomial testing problem and its
variants with two folds of objectives. One is to understand how
this problem relates to critical problems in complexity, and if so
to what extent. The other is to exploit possibilities of applying
algebraic properties of polynomials to the study of those critical
problems.

The paper is organized as follows. We first define $\Pi\Sigma\Pi$
and $\Pi\Sigma$ polynomials. The first is a product of clauses
such that each clause is a sum of terms and each term is a product
of variables. The second is like the first except that each term
is just one variable. These polynomials have easy depth-$3$ or
depth-$2$ circuit representations that have been extensively
studied for the polynomial identity testing problem. We prove a
series of results: The multilinear monomial testing problem for
$\Pi\Sigma\Pi$ polynomials is NP-hard, even when each clause has
at most three terms. The testing problem for $\Pi\Sigma$
polynomials is in P, and so is the testing for two-term
$\Pi\Sigma\Pi$ polynomials. However, the testing for a product of
one two-term $\Pi\Sigma\Pi$ polynomial and another $\Pi\Sigma$
polynomial is NP-hard. This type of polynomial product is, more or
less, related to the polynomial factorization problem. We also
prove that testing $c$-monomials for two-term $\Pi\Sigma\Pi$
polynomials is NP-hard for any $c> 2$, but the same testing is in
P for $\Pi\Sigma$ polynomials. Finally, two parameterized
algorithms was devised for three-term $\Pi\Sigma\Pi$ polynomials
and products of two-term $\Pi\Sigma\Pi$ and $\Pi\Sigma$
polynomials. These results have laid a basis for further study
about testing monomials.

\section{Notations and Definitions}

Let ${\cal P}\in \{Z, Z_N, Z_2\}$, $N>2$. For variables $x_1,
\dots, x_n$, let ${\cal P} [x_1,\cdots,x_n]$ denote the
communicative ring of all the $n$-variate polynomials with
coefficients from ${\cal P}$. For $1\le i_1 < \cdots <i_k \le n$,
$\pi =x_{i_1}^{j_1}\cdots x_{i_k}^{j_k}$ is called a monomial. The
degree of $\pi$, denoted by $\mbox{deg}(\pi)$, is
$\sum^k_{s=1}j_s$. $\pi$ is multilinear, if $j_1 = \cdots = j_k =
1$, i.e., $\pi$ is linear in all its variables $x_{i_1}, \dots,
x_{i_k}$. For any given integer $c\ge 1$, $\pi$ is called a
$c$-monomial, if $1\le j_1, \dots, j_k < c$.

An arithmetic circuit, or circuit for short, is a direct acyclic
graph with $+$ gates of unbounded fan-ins, $\times$ gates of two
fan-ins, and all terminals corresponding to variables. The size,
denoted by $s(n)$, of a circuit with $n$ variables is the number
of gates in it. A circuit is called a formula, if the fan-out of
every gate is at most one, i.e., the underlying direct acyclic
graph is a tree.

By definition, any polynomial $p(x_1,\dots,x_n)$ can be expressed
as a sum of a list of monomials, called the sum-product expansion.
The degree of the polynomial is the largest degree of its
monomials in the expansion. With this expression, it is trivial to
see whether $p(x_1,\dots,x_n)$ has a multilinear monomial, or a
monomial with any given pattern. Unfortunately, this expression is
essentially problematic and infeasible to realize, because a
polynomial may often have exponentially many monomials in its
expansion.

In general, a polynomial $p(x_1,\dots,x_n)$ can be represented by
a circuit or some even simpler structure as defined in the
following. This type of representation is simple and compact and
may have a substantially smaller size, say, polynomially in $n$,
in comparison with the number of all monomials in the sum-product
expansion. The challenge is how to test whether $p(x_1,\dots,x_n)$
has a multilinear monomial or some needed monomial, efficiently
without unfolding it into its sum-product expansion?


\begin{definition}\scrod
Let $p(x_1,\dots,x_n)\in {\cal P}[x_1,\dots,x_n]$ be any given
polynomial. Let $m, s, t\ge 1$ be integers.
\begin{itemize}
\item $p(x_1,\dots,x_n)$ is said to be a $\Pi_m\Sigma_s\Pi_t$
polynomial, if $p(x_1,\dots,x_n)=\prod_{i=1}^t F_i$, $F_i =
\sum_{j=1}^{r_i} X_{ij}$ and $1\le r_i \le s$, and
$\mbox{deg}(X_{ij})\le t$. We call each $F_i$ a clause. Note that
$X_{ij}$ is not a monomial in the sum-product expansion of
$p(x_1,\dots,x_n)$ unless $m=1$. To differentiate this subtlety,
we call $X_{ij}$ a term.

\item In particular, we say $p(x_1,\dots,x_n)$ is a
$\Pi_{m}\Sigma_s$ polynomial, if it is a
 $\Pi_m\Sigma_s\Pi_1$ polynomial. Here, each clause is a linear addition
 of single variables. In other word, each term has degree $1$.

\item When no confusing arises from the context, we use
$\Pi\Sigma\Pi$ and $\Pi\Sigma$ to stand for $\Pi_m\Sigma_s\Pi_t$
and $\Pi_m\Sigma_s$ respectively.

Similarly, we use $\Pi\Sigma_s\Pi$ and $\Pi\Sigma_s$ to stand for
$\Pi_m\Sigma_s\Pi_t$ and $\Pi_m\Sigma_s$ respectively, emphasizing
that every clause in a polynomial has at most $s$ terms or is a
linear addition of at most $s$ single variables.

\item For any given integer $k\ge 1$, $p(x_1,\dots,x_n)$ is called
a $k$-$\Pi\Sigma\Pi$ polynomial, if each of its terms has
$k$ distinct variables.

\item $p(x_1,\dots,x_n)$ is called a $\Pi\Sigma\Pi \times
\Pi\Sigma$ polynomial, if $p(x_1,\dots,x_n) = p_1 p_2$ such that
$p_1$ is a $\Pi\Sigma\Pi$ polynomial and $p_2$ is a $\Pi\Sigma$
polynomial. Similarly, $p(x_1,\dots,x_n)$ is called a
$k$-$\Pi\Sigma\Pi \times \Pi\Sigma$ polynomial, if
$p(x_1,\dots,x_n) = p_1 p_2$ such that $p_1$ is a
$k$-$\Pi\Sigma\Pi$ polynomial and $p_2$ is a $\Pi\Sigma$
polynomial.
\end{itemize}
\end{definition}

It is easy to see that a $\Pi_m\Sigma_s\Pi_t$ or $\Pi_m\Sigma_s$
polynomial may has as many as $s^m$ monomials in its sum-product
expansion.

On the surface, a $\Pi_m\Sigma_s\Pi_t$ polynomial {\em
"resembles"} a SAT formula, especially when $t=1$. Likewise, a
$\Pi_m\Sigma_3\Pi_t$ ($\Pi_m\Sigma_2\Pi_t$) polynomial {\em
"resembles"} a 3SAT (2SAT) formula, especially when $t=1$.
However, negated variables are not involved in a polynomials.
Furthermore, as pointed out in the previous section, it is not
easy, if not impossible, to have some easy algebra to deal with
the properties of $x^2 = x$ and $x \cdot \bar{x} =  0$ in a field,
especially when the field is larger than $Z_2$. Also, as pointed
out before, the arithmetization technique in Shamir
\cite{shamir92} is not applicable to this case.

\section{$\Pi\Sigma\Pi$ Polynomials}

Given any $\Pi_m\Sigma_s\Pi_t$ polynomial $p(x_1,\ldots,x_n) = p_1
\cdots p_m$, one can nondeterministically choose a term $\pi_i$
from the clause $p_i$ and then check whether $\pi_1 \cdots \pi_m$
is a multilinear monomial. So the problem of testing multilinear
monomials in a $\Pi\Sigma\Pi$ polynomial is in NP. In the
following we show that this problem is also NP-hard.

\begin{theorem}\label{mlm-thm}
It is NP-hard to test whether a $2$-$\Pi_m\Sigma_3\Pi_2$
polynomial has a multilinear monomial in its sum-product
expansion.
\end{theorem}

Note that every clause in such a $2$-$\Pi_m\Sigma_3\Pi_2$
polynomial has at most three terms such that each term has at most
two distinct variables.

\begin{proof}
We reduce 3SAT to the given problem. Let $f=f_1 \wedge \cdots
\wedge f_m$  be a 3SAT formula. Without loss of generality, we
assume that every variable $x_i$ in $f$ appears at most three
times, and if $x_i$ appears three times, then $x_i$ itself occurs
twice and $\bar{x}_i$ once. (It is easy to see that a simple
preprocessing procedure can transform any 3SAT formula to satisfy
these properties.)

Let $x_i$ be any given variable in  $f$, we introduce new
variables to replace it. If $x_i$ appears only once then we
replace the appearance of $x_i$ (or $\bar{x}_i$) by a new variable
$y_{i1}$. When $x_i$ appears twice, then we do the following: If
$x_i$ (or its negation $\bar{x}_i$)  occurs twice, then replace
the first occurrence by a new variable $y_{i1}$ and the second by
$y_{i2}$. If both $x_i$ and $\bar{x}_i$ occur, then replace both
occurrences by $y_{i1}$. When $x_i$ occurs three times with $x_i$
appearing twice and $\bar{x}_i$ once, then replace the first
$x_{i}$ by $y_{i1}$ and the second by $y_{i2}$, and replace
$\bar{x}_i$ by $y_{i1}y_{i2}$. This procedure of replacing all
variables in $f$, negated or not, with new variables can be
carried out easily in quadratic time.

Let $p = p_1\cdots p_m$ be polynomial resulting from the above
replacement process. Here, $p_i$ corresponds to  $f_i$ with
boolean literals being replaced. Clearly, $p$ is a
$2$-$\Pi_m\Sigma_3\Pi_2$ polynomial.

We now consider the sum-product expansion of $f = f_1 \cdots f_m$.
It is easy to see that $f$ is satisfiable iff its sum-product
expansion has a product
$$
\psi = \tilde{x}_{i_1} \cdots \tilde{x}_{i_m},
$$
where the literal $\tilde{x}_{i_j}$ is from the clause $f_j$ and
is either $x_{i_j}$ or $\bar{x}_{i_j}$, $1\le j \le m$.
Furthermore,  the negation of $\tilde{x}_{i_j}$ must not occur in
$\pi$.

Let $t(\tilde{x}_{i_j})$ denote the replacement of
$\tilde{x}_{i_j}$ by new variables $y_{i_j1}$ and/or $y_{i_j2}$ as
described above to transform $f$ to $p$. Then,
$t(\tilde{x}_{i_j})$ is a term in the clause $p_j$. Hence,
$$
t(\psi) = t(\tilde{x}_{i_j}) \cdots t(\tilde{x}_{i_m})
$$
is a monomial in the sum-product expansion of $p$. Moreover,
$t(\psi)$ is multilinear, because a variable and its negation
cannot appear in $\pi$ at the same time.

On the other hand, assume that
$$
\pi = \pi_1 \cdots \pi_m
$$
is a multilinear monomial in $p$ with the term $\pi_{i_j}$ in the
clause $p_j$. Let $t^{-1}(\cdot)$ denote the reversal replacement
of $t(\cdot)$. Then, by the procedure of the replacement above,
$t^{-1}(\pi_{i_j})$ is a variable or the negation of a variable in
$f_j$. Thus,
$$
t^{-1}(\pi) = t^{-1}(\pi_1) \cdots t^{-1}(\pi_m)
$$
is a product in the sum-product expansion of $f$. Since $\pi$ is
multilinear,  a variable and its negation cannot appear in
$t^{-1}(\pi)$ at the same time. This implies  that $f$ is
satisfiable by an assignment of setting all the literals in
$t^{-1}(\pi)$ true.
\end{proof}

We give an example to illustrate the variable replacement
procedure given in the above proof.   Given a 3SAT formula
$$
f = (x_1 \vee \bar{x}_2 \vee x_3) \wedge (\bar{x}_1 \vee x_2 \vee
x_4) \wedge  (x_1 \vee x_2 \vee \bar{x}_3) \wedge(x_4 \vee x_5),
$$
the polynomial for $f$ after variable replacements is
$$
p(f) = (y_{11} + y_{21}y_{22} + y_{31})(y_{11}y_{12} + y_{21} +
y_{41})(y_{12} + y_{22} + y_{31}) (y_{42} + y_{51}).
$$
The truth assignment satisfying $f$ as determined by the product
$x_3 \cdot \bar{x}_1 \cdot x_2 \cdot x_4$ is one to one
correspondent to the multilinear monomial $y_{31} \cdot
y_{11}y_{12} \cdot y_{22} \cdot y_{42}$ in $p(f)$.

Two corollaries follow immediately from this theorem.

\begin{corollary}
For any $s\ge 3$, it is NP-hard to test whether a
$\Pi_m\Sigma_s\Pi_t$ polynomial has multilinear monomials in its
sum-product expansion.
\end{corollary}

\begin{corollary}
It is NP-hard to test whether a polynomial has multilinear
monomials in its sum-product expansion, when the polynomial is
represented by a general arithmetic circuit.
\end{corollary}

The NP-hardness in the above corollary was obtained by Koutis
\cite{koutis08}.

\section{$\Pi\Sigma$ Polynomials}

Note that every clause in a $\Pi\Sigma$ polynomial $p$ is a linear
addition of single variables. $p$ looks very much like a SAT
formula. But this kind of structural "resemblance" is very
superficial, as we will show in the following that the multilinear
monomial testing problem for $p$ is in P. This shows that terms
with single variables do not have the same expression power as
boolean variables and their negations together can achieve. As
exhibited in the proof of Theorem \ref{mlm-thm}, terms with two
variables are equally powerful as boolean variables together with
their negations. Hence, it is interesting to see that a complexity
boundary exists between polynomials with terms of degree $1$ and
those with terms of degree $2$.

\begin{theorem}\label{linear-thm}
There is a $O(ms\sqrt{m+n})$ time algorithm to test if a $\Pi_m\Sigma_s$ polynomial  has a multilinear monomial
in its sum-product expansion.
\end{theorem}

\begin{proof}
Let $f(x_1,\ldots,x_n)=f_1 \ldots f_m$ be any given
$\Pi_m\Sigma_s$ polynomial. Without loss of generality, we assume
that  each clause has exactly $s$ many terms, i.e., $f_i =
\sum^{s}_{j=1} x_{ij}$, $1\le i\le s$. We shall reduce the problem
of testing multilinear monomials in $f(x_1,\ldots,x_n)$ to the
problem of finding a maximum matching in some bipartite graph.

We construct a bipartite graph $G=(V_1 \cup V_2, E)$ as follows.
$V_1 = \{v_1,\ldots,v_m\}$ so that each $v_i$ represents the
clause $f_i$. $V_2 = \{x_1,\ldots, x_n\}$. For each clause $f_i$,
if it contains a variable $x_j$ then we add an edge $(v_i, x_j)$
into $E$.

Suppose that $f(x_1,\ldots,x_n)$ has a multilinear monomial
$$
\pi = x_{i_1} \cdots x_{i_m}
$$
with $x_{i_j}$ in $f_j$, $1\le j\le m$. Then, all the variables in
$\pi$ are distinct. Thus, we have a maximum matching of size $m$
$$
(v_1, x_{i_1}), \ldots, (v_m, x_{i_m}).
$$

Now, assume that we have a maximum matching of size $m$
$$
(v_1, x'_{i_1}), \ldots, (v_1, x'_{i_m}).
$$
Then, all the variables in the matching are distinct. Moreover, by
the construction of the graph $G$, $x'_{i_j}$ are in the clause
$f_j$, $1\le j\le m$. Hence,
$$
\pi' = x'_{i_1} \cdots x'_{i_m}
$$
is a multilinear monomial in $f(x_1,\ldots,x_n)$

It is well-known that finding a maximum matching in a bipartite
graph can be done in $O(|E|\sqrt{|V|})$ time
\cite{aspvall-plass-tarjan79}. So the above reduction shows that
we can test whether $f(x_1,\ldots,x_n)$ has a multilinear monomial
in $O(ms\sqrt{m+n})$, since the graph $G$ has $m+n$ vertices and
at most $ms$ edges.
\end{proof}

In the following, we give an extension of Theorem
\ref{linear-thm}.

\begin{theorem}\label{linear-thm2}
There is a $O(tc^k ms\sqrt{m+n})$ time algorithm to test whether any given
$\Pi_k\Sigma_c\Pi_t\times \Pi_m\Sigma_s$ polynomial has a multilinear monomial in its
sum-product expansion.
\end{theorem}

\begin{proof}
Let $p=p_1p_2$ be any given $\Pi_k\Sigma_c\Pi_t\times
\Pi_m\Sigma_s$ polynomial such that $p_1 = f_1\cdots f_k$ is a
$\Pi_k\Sigma_c\Pi_t$ polynomial and $p_2 = g_1\cdots g_m$ is a
$\Pi_m\Sigma_s$ polynomial. Note that every clause $f_i$ in $p_1$
has at most c terms with degree at most $t$. So, $p_1$ has at most
$c^k$ products in its sum-product expansion. Hence, in $O(tc^k)$
time, we can list all the products in that expansion, and let
$\cal{C}$ denote the set of all those products.

It is obvious that $p$ has a multilinear monomial, iff there is one product $\psi \in \cal{C}$ such that
the polynomial $\psi p_2$ has a multilinear monomial.

Now, for any product $\psi \in \cal{C}$,  we consider how to test whether the polynomial
$$
p(\psi) = \psi \cdot p_2 = \psi \cdot g_1 \cdots g_m
$$
have a multilinear polynomial. Let
$$
\pi = \psi \cdot \pi_1 \cdots \pi_m
$$
be an arbitrary product in the sum-product expand of $p(\psi)$
with the term $\pi_i$ in $g_i$, $1\le i\le m$. Since $\psi$ is
fixed, in order to make $\pi$ to be multilinear, each $\pi_i$ must
not have a variable in $\psi$. This observation helps us devise a
one-pass {\em "purging"} process to eliminate all the variables in
every clause  of $g_i$ that cannot be included in a multilinear
monomial in $p(\psi)$. The purging works as follows: For each
clause $g_i$, eliminate all its variables that also appear in
$\psi$. Let $g_i'$ be the resulting clause of $g_i$, and $p'_2 =
\cdot g_1' \cdots g_m'$ be the resulting polynomial of $p_2$. If
any $g_i'$ is empty, then there is no multilinear monomials in
$\psi \cdot p'_2$, hence no multilinear monomials in $p(\psi)$.
Otherwise, by Theorem \ref{linear-thm}, we can decide whether
$p_2'$ has a multilinear monomial, hence whether $p(\psi)$ has a
multilinear monomial,  in $O(ms\sqrt{m+n})$ time.

Putting all the steps together, we can test whether $p$ has a multilinear monomial in
$O(tc^k ms\sqrt{m+n})$ time.
\end{proof}

\section{$\Pi\Sigma_2\Pi$ polynomials}

In Section 3, we has proved that the multilinear monomial testing
problem for any $\Pi\Sigma_s\Pi$ polynomials with at most $s\ge 3$
terms in each clause is NP-hard. In this section, we shall show
that another complexity boundary exists between $\Pi\Sigma_3\Pi$
polynomials and $\Pi\Sigma_2\Pi$ polynomials. As noted before, a
$\Pi\Sigma_2\Pi$ polynomial may look like a 2SAT formula, but they
are essentially different from each other. For example, unlike
2SAT formulas, no implication can be derived for two terms in a
clause. Thus, the classical algorithm based on implication graphs
for 2SAT formulas by Aspvall, Plass and Tarjan
\cite{aspvall-plass-tarjan79} does not apply to $\Pi\Sigma_2\Pi$
polynomials. The implication graphs can also help prove that 2SAT
is NL-complete \cite{papadimitrious94}. But we do not know whether
the monomial testing problem for $\Pi\Sigma_2\Pi$ polynomials is
NL-complete or not. We feel that it may be not. There is another
algorithm for solving 2SAT in quadratic time via repeatedly {\em
"purging"} contradicting literals. The algorithm devised in the
following more or less follows a similar approach of that
quadratic time algorithm.

\begin{theorem}\label{2-term-thm}
There is a quadratic time algorithm to test whether any given
$\Pi_m\Sigma_2\Pi_t$ polynomial has a multilinear monomial in its
sum-product expansion.
\end{theorem}

\begin{proof}
Let $f = f_1\cdots f_m$ be any given $\Pi_m\Sigma_2\Pi_t$ polynomial such that
$f_i = (T_{i1}+T_{i2})$ and each term has degree at most $t$. Let
$$
\pi = \pi_1 \cdots \pi_m
$$
be any monomial in the sum-product expansion of $f$. Here term
$\pi$ is either $T_{i1}$ or $T_{i2}$, $1\le i\le m$. Observe that
$\pi$ is multilinear, iff any two terms in it must not share a
common variable. We now devise a {\em "purging"} based algorithm
to decide whether a multilinear monomial $\pi$ exists in $f$. The
purging part of this algorithm is similar to what is used in the
proof of Theorem \ref{linear-thm2}.

The purging algorithm works as follows. We select any clause $f_i$
from $f$, and choose a term  in $f_i$ for $\pi_i$. we purge all
the terms in the remaining clauses that share a common variable
with $\pi_i$. Once we find one clause with one term being purged
but with the other left, we then choose this remaining term in
that clause to repeat the purging process.

The purging stops for $\pi_i$ when one of the three possible
scenarios happens:

(1) We find one clause $f_j$ with two terms being purged. In this
case, any of the two terms in $f_j$ cannot be chosen to form a
multilinear monomial along with $\pi_i$. So, we have to choose the
other term in $f_i$ for $\pi$, if that term has not been chosen.
We use this $\pi_i$ to repeat the same purging process. If $f_i$
has not term left, then this means that neither term in $f_i$ can
be chosen to form a multilinear monomial, so the answer is {\em
"NO"}.

(2) We find that every clause $f_j$ contributes one term $\pi_j$
during the purging process. This means that $\pi = \pi_1 \cdots
\pi_m$ has no variables appearing more than once, hence it is a
multilinear monomial, so an answer {\em "YES"} is obtained.

(3) We find that the purging process fails to purge any terms in a
subset of clauses. Let $S\subset I$ denote the index of these
clause, where $I = \{1, \ldots, m\}$. Let $\pi'$ be the product of
$\pi_j$ with $j \in I - S$. According to the purging process,
$\pi'$ does not share any common variables with terms in any
clause $f_u$ with $u\in S$. Hence, the input polynomial $f$ has a
multilinear monomial iff the product of those clauses $f_u$ has a
multilinear monomial. Therefore, we recursively to apply the
purging process to this product of clauses. Note that this product
has at least one fewer clause than $f$.

With the help of some simple data structure, the purging process can be implemented in quadratic time.
\end{proof}

\section{$\Pi\Sigma_2\Pi\times \Pi\Sigma$ Polynomials vs.
$\Pi\Sigma_2\Pi$ and $\Pi\Sigma$ Polynomials}

In structure, a $\Pi\Sigma_2\Pi\times \Pi\Sigma$ polynomial is a
product of one $\Pi\Sigma_2\Pi$ polynomial and another $\Pi\Sigma$
polynomial. This structural characteristic is somehow related to
polynomial factorization. It has been shown in Sections 4 and 5
that testing multilinear monomials in $\Pi\Sigma_2\Pi$ or
$\Pi\Sigma$ polynomials can be done respectively in polynomial
time. This might encourage one to think that testing multilinear
monomials in $\Pi\Sigma_2\Pi\times \Pi\Sigma$ polynomials could
also be done in polynomial time. However, a little bit
surprisingly the following theorem shows that a complexity
boundary exists, separating $\Pi\Sigma_2\Pi\times \Pi\Sigma$
polynomials from $\Pi\Sigma_2\Pi$ and $\Pi\Sigma$ polynomials.

\begin{theorem}\label{product-thm}
The problem of testing multilinear monomials in $\Pi\Sigma_2\Pi\times \Pi\Sigma$
polynomials is NP-complete.
\end{theorem}

\begin{proof}
It is easy to see that the given problem is in NP. To show that
the problem is also  NP-hard, we consider any given
$\Pi_m\Sigma_3\Pi_t$ polynomial $f = f_1 \cdots f_m$ with $m\ge 1$
and $t\ge 2$ such that each clause $f_i = (T_{i1}+T_{i2}+T_{i3})$
and each term $T_{ij}$ has degree at most $t$, $1\le i\le m$,
$1\le j\le 3$. We shall reduce $f$ into a $\Pi\Sigma_2\Pi\times
\Pi\Sigma$ polynomial. Once this is done, the NP-hardness of the
given problem follows from Theorem \ref{mlm-thm}.

We consider the clause
$$
f_i = (T_{i1}+T_{i2}+T_{i3}).
$$
We want to represent $f_i$ by a $\Pi\Sigma_2\Pi\times \Pi\Sigma$
polynomial so that selecting exactly one term  from $f_i$ is
equivalent to selecting exactly one monomial from the new
polynomial with exactly one term $T_{ij}$ in $f_i$ under the
constraint that the newly introduced variables are linear in the
monomial. We construct the new polynomial, denoted by $p(f_i)$, as
follows.
$$
p(f_i) = (T_{i1}u_i + v_i)(T_{i2}u_i + w_i)(T_{i3}u_i + z_i)(v_i + w_i + z_i),
$$
where $u_i, v_i, w_i$ and $z_i$ are new variables.
It is easy to see that there are only three monomials in $p(f_i)$ satisfying the constraint:
$$
T_{i1} u_i v_i w_i z_i,~~ T_{i2} u_i v_i w_i z_i, \mbox{~~and}~~
T_{i3} u_i v_i w_i z_i.
$$
Each of those three monomials corresponds to exactly one term in $f_i$.
Now, let
$$
p(f) = p(f_1)\cdots p(f_m)
$$
be the new polynomial representing $f$ and
$$
\pi = \pi_1 \cdots \pi_m
$$
be a monomial in $f$ with terms $\pi_i$ in $f_i$. If $\pi$ is
multilinear, then so is
$$
\pi' = (\pi_1 u_1 v_1 w_1 z_1)\cdots (\pi_m u_m v_m w_m z_m)
$$
in $p(f)$. On  the other hand, if
$$
\psi = \psi_1 \cdots \psi_m
$$
is multilinear monomial in $p(f)$, then $\psi_i = T_{ij_i} u_i v_i w_i z_i$ with $j_i \in \{1, 2, 3\}$.
This implies that
$$
\psi' = T_{1j_1} \cdots T_{mj_m}
$$
must be a multilinear monomial in $f$. Obviously, the reduction
from $f$ to $p(f)$ can be done in polynomial time.
\end{proof}

\section{Testing $c$-Monomials}\label{c-monomials}

By definition, a multilinear monomial is a $2$-monomial. It has
been shown in Section 5 that the problem of testing multilinear
monomials in a $\Pi\Sigma_2\Pi$ polynomial is solvable in
quadratic time. We shall  show that  another complexity boundary
exists to separate $c$-monomials from $1$-monomials, even when
$c=3$. On the positive side, we shall show that it is efficient to
testing $c$-monomials for $\Pi\Sigma$ polynomials.

\begin{theorem}
The problem of testing $3$-monomials in any
$3$-$\Pi_m\Sigma_2\Pi_6$ polynomial is NP-complete.
\end{theorem}

\begin{proof}
We only need to show that the problem is NP-hard, since it is
trivial to see that the problem is in NP.

Let $f = f_1 \cdots f_m$ be any given $2$-$\Pi_m\Sigma_3\Pi_2$
polynomial, where each clause $f_i = (T_{i1}+T_{i2}+T_{i3})$ and
each term $T_{ij}$ is multilinear with at most $2$ distinct
variables, $1\le i\le m$, $1\le j\le 3$. By Theorem \ref{mlm-thm},
testing whether $f$ has a multilinear monomial is NP-hard. We now
show how to construct a $3$-$\Pi_m\Sigma_2\Pi_6$ polynomial to
represent $f$ with the property that $p$ has a multilinear
monomial iff the new polynomial has a $2$-monomial.

We consider the clause
$$
f_i = (T_{i1}+T_{i2}+T_{i3}).
$$
We want to represent $f_i$ by a $3$-$\Pi\Sigma_2\Pi_6$ polynomial
so that selecting exactly one term  from $f_i$ is equivalent to
selecting exactly one $2$-monomial from the new polynomial with
exactly one term $T_{ij}$ in $f_i$ under the constraints that
$T_{ij}$ appears twice and the newly introduced variables are each
of degree $2$. The idea for constructing the new polynomial seems
like what is used in the proof of Theorem \ref{product-thm}, but
it is different from that construction.  We design the new
polynomial, denoted by $p(f_i)$, as follows.
$$
p(f_i) = (T_{i1} T_{i1} u_i^2 + v_i)(T_{i2} T_{i2} u_i^2 + v_i)(T_{i3} T_{i3} u_i^2 + v_i)
$$
where $u_i$ and  $v_i$ are new variables. Since each term $T_{ij}$
is multilinear with at most two distinct variables, $p(f_i)$ is a
$3$-$\Pi_m\Sigma_2\Pi_6$ polynomial. It is easy to see that there
are no multilinear monomials in $p(f_i)$. But there are three
monomials in $p(f_i)$ satisfying the given constraints:
$$
T_{i1} T_{i1} u_i^2 v_i^2,~~ T_{i2} T_{i2} u_i^2 v_i^2,\
\mbox{~and~} \ T_{i3} T_{i3} u_i^2 v_i^2.
$$
Each of those three monomials corresponds to exactly one term in
$f_i$. Note that only those three monomials in $p(f_i)$ can
possibly be $3$-monomials, depending on whether  $T_{ij}T_{ij}$ is
a $3$-monomials.  Now, let
$$
p(f) = p(f_i)\cdots p(f_m)
$$
be the new polynomial representing $f$ and
$$
\pi = \pi_1 \cdots \pi_m
$$
be a monomial in $f$ with terms $\pi_i$ in $f_i$. If $\pi$ is
multilinear, then
$$
\pi' = (\pi_1 \pi_1 u_1^2 v_1^2)\cdots (\pi_m \pi_m u_m^2 v_m^2)
$$
is a $3$-monomial in $p(f)$. On  the other hand, if
$$
\psi = \psi_1 \cdots \psi_m
$$
is a $3$-monomial in $p(f)$, then $\psi_i = T_{ij_i} T_{ij_i}
u_i^2 v_i^2$ with $j_i \in \{1, 2, 3\}$. This implies that
$$
\psi' = T_{1j_1} T_{1j_1} \cdots T_{mj_m}T_{mj_m}
$$
is a $3$-monomial. Therefore,
$$
\psi'' = T_{1j_1}  \cdots T_{mj_m}
$$
must be a multilinear monomial in $f$. Obviously, reducing $f$  to
$p(f)$ can be done in polynomial time.
\end{proof}

The following corollaries follows immediately from Theorem
\ref{c-monomials}:

\begin{corollary}
For any $c > 2$, testing $c$-monomials in any $\Pi_m\Sigma_s\Pi_t$
polynomial is NP-complete.
\end{corollary}

\begin{corollary}
For any $c> 2$, testing $c$-monomials in any $\Pi_m\Sigma_s\Pi_t$
polynomial represented by a formula or a general arithmetic
circuit is NP-complete.
\end{corollary}

Recall that by Theorem \ref{linear-thm} the  multilinear monomial
testing problem for $\Pi\Sigma$ polynomials is solvable in
polynomial time. The following theorem shows a complementary
result about $c$-monomial testing for the same type of
polynomials.

\begin{theorem} There is a $O(cms\sqrt{m+cn})$ time algorithm to test
whether any $\Pi_m\Sigma_s$ polynomial has a $c$-monomial or not,
where $c>2$ is a fixed constant.
\end{theorem}

\begin{proof}
We consider to generalize the maximum matching reduction in
Theorem \ref{linear-thm}. Like before, Let $f(x_1,\ldots,x_n)=f_1
\ldots f_m$ be any given $\Pi_m\Sigma_s$ polynomial such that $f_i
= \sum^{s}_{j=1} x_{i_j}$, $1\le i\le s$. We construct a bipartite
graph $G=(V_1 \cup V_2, E)$ as follows. $V_1 = \{v_1,\ldots,v_m\}$
so that each $v_i$ represents the clause $f_i$. $V_2 =
\cup^n_{i=1} \{u_{i1}, u_{i2}, \ldots, u_{i(c-1)}\}$, i.e., each
variable $x_i$ corresponds to $c-1$ vertices $u_{i1}, u_{i2},
\ldots, u_{i(c-1)}$. For each clause $f_i$, if it contains a
variable $x_j$ then we add $c-1$ edges $(v_i, u_{jt})$ into $E$,
$1\le t\le c-1$.

Suppose that $f(x_1,\ldots,x_n)$ has a $c$-monomial
$$
\pi = x_{i_1} \cdots x_{i_m}
$$
with $x_{i_j}$ in $f_j$, $1\le j\le m$. Note that each variable
$x_{i_j}$ appears $k(x_{i_j}) < c$ times in $\pi$. Those
appearances correspond to $k(x_{i_j})$ clauses $f_{t_1}, \ldots,
f_{t_k(x_{i_j})}$ from which $x_{i_j}$ was respectively selected
to form $\pi$. This implies that there are $k(x_{i_j})$ edges
matching $v_{t_1}, \ldots, v_{t_k(x_{i_j})}$ with $k(x_{i_j})$
vertices in $V_2$ that represent $x_{i_j}$. Hence, the collection
of $m$ edges for $m$ appearances of all the variables, repeated or
not, in $\pi$ forms a maximum matching of size $m$ in the graph G.

Now, assume that we have a maximum matching of size $m$
$$
(v_1, u_{i_1j_1}), \ldots, (v_m, u_{i_m j_m}).
$$
Recall that $u_{i_t j_t}$, $1\le t \le m$, is designed to
represent the variable $x_{i_t}$. By the construction of the graph
$G$, $x_{i_t}$ are in the clause $f_t$, $1\le t\le m$, and it may
appear $c-1$ times. Hence,
$$
\pi = x_{i_1} \cdots x_{i_m}
$$
is a $c$-monomial in $f(x_1,\ldots,x_n).$

With the help of the $O(|E|\sqrt{|V|})$ time algorithm
\cite{hopcroft-karp73} for finding a maximum matching in a
bipartite graph, testing whether $f(x_1,\ldots,x_n)$ has a
 $c$-monomial can done in $O(cms\sqrt{m+cn})$,
 since the graph $G$ has $m+cn$ vertices and at most $cms$ edges.
\end{proof}

\section{Parameterized Algorithms}

In this section, we shall devise two parameterized algorithms for
testing multilinear monomials in  $\Pi_m\Sigma_3\Pi_t$ and
$\Pi_m\Sigma_2\Pi_t \times \Pi_k\Sigma_3$ polynomials. By Theorems
\ref{mlm-thm} and \ref{product-thm}, the multilinear monomial
testing problem for each of these two types of polynomials is
NP-complete.

\begin{theorem}
There is a $O(tm^2 1.7751^m)$ time algorithm to test whether any
 $\Pi_m\Sigma_3\Pi_t$ polynomial has a multilinear monomial in its sum-product expansion.
\end{theorem}

\begin{proof}
Let $f = f_1 \cdots f_m$ be any given  $\Pi_m\Sigma_3\Pi_t$ polynomial, where each clause
$$
f_i = (T_{i1}+T_{i2}+T_{i3})
$$
and each term $T_{ij}$ has degree at most $t$, $1 \le i\le m$, $1
\le j\le 3$.

We now consider to reduce $f$ to an undirected graph $G = (V, E)$
such that $f$ has a multilinear monomial iff $G$ has a maximum
$m$-clique. For each clause $f_i$, we design three vertices
$v_{i1}, v_{i2}$ and $v_{i3}$, representing the three
corresponding terms in $f_i$. Let $V$ be the collection of those
vertices for all the terms in $f$. For any two vertices $v_{ij}$
and $v_{i'j'}$ with $ i\not= i'$, we add an edge $(v_{ij},
v_{i'j'})$ to $E$, if their corresponding terms $T_{ij}$ and
$T_{i'j'}$ do not share any common variable. Since any two
vertices designed for the terms in a clause are not connected, the
maximum cliques in $G$ could have $m$ vertices corresponding to
$m$ terms, each of which is in one of those $m$ clauses. Let
$$
\pi = \pi_1 \cdots \pi_m
$$
be any monomial in $f$ with $\pi$ being a term from $f_i$. We
consider two cases in the following.

Assume that $\pi$ is multilinear monomial. Let $\pi_i = T_{ij_i}$,
$j_i \in \{1, 2, 3\}$. Then, any two terms $T_{ij_i}$ and
$T_{i'j_{i'}}$ in $\pi$ do not share any common variable. So,
there is an edge $(v_{ij_i}, v_{i'j_{i'}})$ in $E$. Hence, the
graph $G$ has an $m$-clique $\{v_{1j_1}, \ldots, v_{mj_m}\}$.
Certainly, this clique is maximum.

Now, suppose that $G$ has a maximum clique $\{v_{1j_1}, \ldots,
v_{mj_m}\}$. Then, by the construction of $G$, each vertex
$v_{ij_i}$ corresponds to the term $T_{ij_i}$ in the clause $f_i$.
Thus, the product of those $m$ terms is a multilinear monomial,
because any two of those terms do not share a common variable.

Finally, we use Robson's  $O(1.2108^{|V|})$ algorithm to find a
maximum clique for $G$. If the clique has size $m$, then $f$ has a
multilinear monomial. Otherwise, it does not. Note that $|V| =
3m$. Combining the reduction time with the clique finding time
gives an overall $O(tm^2 1.7751^m)$ time.
\end{proof}

We now turn to $\Pi_m\Sigma_2\Pi_t \times \Pi_k\Sigma_3$
polynomials and give the second  parameterized algorithm for this
type of polynomials.

\begin{theorem}\label{2k-fixed-thm}
There is a $O((mk)^2 3^k)$ time algorithm to test whether  any
$\Pi_m\Sigma_2\Pi_t \times \Pi_k\Sigma_3$ polynomial has a
multilinear monomial in its sum-product expansion.
\end{theorem}

\begin{proof}
Let $p = p_1 \cdot p_2$ such that $p_1$ is a $\Pi_m\Sigma_2\Pi_t$
polynomial and $p_2$ is a $\Pi_k\Sigma_3$ polynomial. In $O(3^k)$
time, we list all the products in the sum-product expansion of
$p_2$. Let ${\cal C}$ be the collection of those products. It is
obvious that $p$ has a multilinear monomial iff there is a product
$\pi \in {\cal C}$ such that $p_1 \cdot \pi$ has a multilinear
monomial. Note that $p_1 \cdot \pi$ is a $\Pi_(m+1)\Sigma_2\Pi_t$
polynomial. By Theorem \ref{2-term-thm}, the multilinear monomial
testing problem for $p_1 \cdot \pi$ can be solved by a quadratic
time algorithm. Hence, the theorem follows by applying that
algorithm to $p_1 \cdot \pi$ for every $\pi\in {\cal C}$ to see if
one of them has a multilinear monomial or not.
\end{proof}

\section*{Acknowledgments}

We thank Yang Liu and Robbie Schweller for many valuable
discussions during our weekly seminar. Conversations with them
help inspire us to develop this study of testing monomials. We
thank Yang Liu for presenting Koutis' paper \cite{koutis08} at the
seminar. The $O((ms)^2 3^k)$ upper bound given in Theorem
\ref{2k-fixed-thm} has been improved by Yang Liu to $O((ms)^2
2^k)$.

Bin Fu's research is support by an NSF CAREER Award, 2009 April 1 to 2014 March 31.


\begin{thebibliography}{99}
\bibliographystyle{plain}

\bibitem{aks04}
A. Manindra, K. Neeraj, and S. Nitin, PRIMES is in P, Ann. of
Math, 160(2): 781-793, 2004.

\bibitem{arora98}
S. Arora, C. Lund, R. Motwani, M. Sudan, and M. Szegedy, Proof
verification and the hardness of approximation problems, Journal
of the ACM 45  (3): 501–555, 1998.

\bibitem{aspvall-plass-tarjan79}
Bengt Aspvall, Michael F. Plass and Robert E. Tarjan,
A linear-time algorithm for testing the truth of certain quantified boolean formulas,
Information Processing Letters  8 (3): 121-123, 1979.

\bibitem{beigel93}
Richard Beigel, The polynomial method in circuit compplexity,
Proceedings of the Eighth Conference on Structure in Complexity
Theory, pp. 82-95, 1993.

\bibitem{bshouty95}
Nader H. Bshouty, Zhixiang Chen, Scott E. Decatur, and Steve Homer,
One the learnability of $Z_N$-DNF formulas,
Proceedings of the Eighth Annual Conference on Computational Learning Theory (COLT 1995),
Santa Cruz, California, USA. ACM, 1995, pp. 198-205.

\bibitem{chen00}
Zhi-Zhong Chen and Ming-Yang Kao, Reducing randomness via
irrational numbers, SIAM J. Comput. 29(4): 1247-1256, 2000.

\bibitem{feige96}
U. Feige, S. Goldwasser, L. Lov\'asz, S. Safra, and M. Szegedy,
Interactive proofs and the hardness of approximating cliques,
Journal of the ACM (ACM) 43 (2): 268–292, 1996.

\bibitem{fu92}
Bin Fu, Separating PH from PP by relativization, Acta Math. Sinica
8(3):329-336, 1992.

\bibitem{hopcroft-karp73}
John E. Hopcroft and Richard M. Karp,
An $n^{5/2}$  algorithm for maximum matchings in bipartite graphs,
SIAM Journal on Computing 2 (4): 225-231, 1973.

\bibitem{kabanets03}
V. Kabanets and R. Impagliazzo, Derandomizing polynomial identity
tests means proving circuit lower bounds, STOC, pp. 355-364, 2003.

\bibitem{klivans01}
Adam Klivans and Rocco A. Servedio, Learning DNF in time
$2^{\tilde{O}(n^{1/3})}$, STOC, pp. 258-265, 2001.

\bibitem{koutis08}
Ioannis Koutis, Faster algebraic algorithms for path and packing problems,
Proceedings of the International Colloquium on Automata,
Language and Programming (ICALP), LNCS, vol. 5125, Springer, pp. 575-586, 2008.

\bibitem{minsky-papert68}
M. Minsky and S. Papert, Perceptrons (expanded edition 1988), MIT
Press, 1968.

\bibitem{papadimitrious94}
Christos H. Papadimitriou, Computational Complexity, Addison-Wesley, 1994.

\bibitem{robson86}
J. M. Robson,
Algorithms for maximum independent sets, Journal of Algorithms 7  (3): 425-440, 1986.

\bibitem{shamir92}
A. Shamir, IP = PSPACE, Journal of the ACM, 39(4): 869-877, 1992.


\bibitem{williams09}
Ryan Williams, Finding paths of length $k$ in $O^*(2^k)$ time, Information Processing Letters, 109, 315-318, 2009.
\end{thebibliography}
\end{document}